\documentclass[a4paper,USenglish]{lipics}

\usepackage[utf8]{inputenc}
\inputencoding{latin2}

\theoremstyle{plain}
\newtheorem {proposition} [theorem] {Proposition}

\usepackage{tabu} 
\usepackage{paralist} 
\usepackage{amsmath}
\usepackage{amssymb}
\usepackage{stmaryrd}
\usepackage{cases} 
\usepackage{relsize} 

\usepackage{tikz}
\usetikzlibrary{automata,arrows,patterns}

\usepackage{sty/comments}
\usepackage{sty/mathematics}
\usepackage{sty/languages}
\usepackage[super,negative]{sty/smnth}

\begin{document}

\title{Munchausen Iteration}

\author{Roland Meyer}
\author{Sebastian Muskalla}

\affil
{
    TU Kaiserslautern, \texttt{\{meyer, muskalla\}@cs.uni-kl.de}
}

\authorrunning{R. Meyer, and S. Muskalla} 

\Copyright{Roland Meyer, and Sebastian Muskalla}


\maketitle

\begin{abstract}
We present a method for solving polynomial equations over idempotent $\omega$-continuous semirings. 
The idea is to iterate over the semiring of functions rather than the semiring of interest, and only evaluate when needed. 
The key operation is substitution.    
In the initial step, we compute a linear completion of the system of equations that exhaustively inserts the equations into one another.  
With functions as approximants, the following steps insert the current approximant into itself. 
Since the iteration improves its precision by substitution rather than computation we named it Munchausen, after the fictional baron that pulled himself out of a swamp by his own hair. 
The first result shows that an evaluation of the $\nth{n}$ Munchausen approximant coincides with the $\nth{2^n}$ Newton approximant. 
Second, we show how to compute linear completions with standard techniques from automata theory.
In particular, we are not bound to (but can use) the notion of differentials prominent in Newton iteration. 
\end{abstract}


\section{Introduction}
Verification problems pop up in a fascinating variety of applications.  
Despite this variety, they are often formulated in a uniform way, as finding the least solution to a system of polynomial equations that is interpreted over a semiring.   
The verification tasks that can be captured this way range from language-theoretic problems underlying model checking~\cite{HeizmannHoenickePodelski2010} to dataflow analyses~\cite{Khedker2009} needed in compilers~\cite{SeidlWilhelmHack2012}.  
The programs that can be handled may involve recursion~\cite{SharirPnueli1978,RepsHorwitzSagiv1995} and weak forms of parallelism~\cite{QadeerRehof2005,HolikMeyer2015}.  
Technically, the system of equations captures the flow of control in the program of interest. 
The semiring interpretation models the aspects of the program semantics that influence the verification task. 
The least solution is the most precise semantic information (of the form one has chosen to track) that is invariant under the program commands. 

Computing the least solution to a given system of equations is an algorithmic challenge, commonly referred to as state space explosion.   
The least solution is the least fixed point of the right-hand side functions. 
The predominant method for computing this fixed point is Kleene iteration. 
In its plain form, Kleene iteration understands the right-hand side functions as a single function over the product domain and applies it over and over again until the least fixed point is reached. 
The practical importance of Kleene iteration stems from the fact that it is amenable to algorithmic optimizations. 
In particular, rather than working on the product domain, implementations use a worklist that only stores the functions whose variables have received updates~\cite{ChaoticIteration1977,NNH1999}. 

Recently, Esparza et al. proposed Newton iteration~\cite{EsparzaKieferLuttenberger2010}, a new method for computing least fixed points that combines Kleene iteration with an acceleration principle (see also~\cite{HopkinsKozen1999,EtessamiYannakakis2009} for precursors of the method).
The idea is indeed inspired by the method for finding roots of numerical functions. 
The current approximant is not only modified by an application of the function, like in Kleene iteration, but in addition shifted towards the fixed point by an acceleration that makes use of the function's differential. 
Newton iteration is guaranteed to converge to the least fixed point, and to do so faster than Kleene iteration (there are even cases where Newton reaches the fixed point while Kleene does not).
On the downside, the Newton steps are computationally more expensive than Kleene steps.
In particular, Newton iteration does not yet decompose into a worklist procedure. 
There has, however, been recent interest in the algorithmics of the method~\cite{FPSolve2014,Reps2016}.

Our contribution is a new iteration scheme for solving systems of polynomial equations $\vec x = \vec f(\vec x)+\vec a$ over semirings.
To be precise, we work with idempotent and $\omega$-continuous semirings, assumptions that are typically met in verification.
Our iteration is exponentially faster than Newton and (arguably) easier to compute. 
To explain the idea, note that both, Kleene and Newton, compute in the semiring of interest. 
Our method works symbolically, over the semiring of functions.
The key idea is substitution. 
In the initial step, we compute a so-called linear completion of the system of equations, $\munchit{0} \ =\ \completionof{\vec f}$. The completion exhaustively inserts the right-hand side functions into each other.
This means an occurrence of variable $y$ in $\vec f_x$ is replaced by $\vec f_y$.
For the resulting function to remain representable, we restrict ourselves to a completion process that is linear in the sense that the next replacement will be done on a variable in $\vec f_y$.
In the iteration step, we continue to insert the current approximant into itself, $\munchit{n+1}\ =\ \evalof{\munchit{n}}{\munchit{n}}$. 
We obtain semantic and algorithmic results about Munchausen iteration.

Concerning the semantics, we show that the Munchausen sequence is faster than the Newton sequence: When we evaluate the $\nth{n}$ Munchausen approximant at the constant vector $\vec a$, we obtain the $\nth{2^n}$ Newton approximant. 
This precise correspondence allows us to transfer deep results from~\cite{EsparzaKieferLuttenberger2010}: The Munchausen sequence converges to the least fixed point and, in the commutative case, is guaranteed to reach the least fixed point in a number of steps that is logarithmic in the number of variables. 
As second main result, we show that there is some flexibility in where to evaluate the approximants.
Any $\vec b$ chosen between $\vec a$ and the least fixed point will guarantee convergence to the least fixed point.
Munchausen iteration thus combines well with further accelerations. 

Concerning the algorithmics, we study the operations of linear completion $\munchit{0} \ =\ \completionof{\vec f}$ and evaluation $\munchit{n+1}\ =\ \evalof{\munchit{n}}{\munchit{n}}$
We show that the linear completion of a function is a linear context-free language. 
Moreover, this language can be represented symbolically by a regular expression, provided the semiring has associated a suitable tensor operation. 
The idea of using tensors was introduced recently in the context of Newton iteration~\cite{Reps2016}.
Our contribution is to lift it to the semiring of functions.
As a second result, we show how to compute the evaluation on the symbolic representation of the linear completion (a linear context-free grammar or even a regular expression). 
The main finding is that the iteration steps are well-behaved to an extent that we can implement them by an indexed language.\\[0.2cm]
{\bf \sffamily Outline}\quad Section~\ref{Section:Preliminaries} introduces the basics on semirings.
The Munchausen iteration scheme is presented in Section~\ref{Section:Acceleration}, together with the study of its semantic properties.
The algorithmics of our iteration is the subject of Section~\ref{Section:Algorithmics}.
Section~\ref{Section:Discussion} concludes the paper with a discussion on the implementation of the method.\\[0.2cm]
Proofs missing in the paper can be found in the appendix.


\section{Systems of Polynomial Equations over $\omega$-Continuous Semirings}\label{Section:Preliminaries}

We study systems of polynomial equations over idempotent and $\omega$-continuous semirings.
A \bfemph{semiring} $\sr$ is a tuple $(\sr, \sradd, \srmult, 0, 1)$ with the following properties (for all $a, b, c\in \sr$):
\begin{align*}
    (\sr, \sradd, 0)
    \text{ is a commutative monoid}
    \quad\quad\quad\quad
    & (\sr, \srmult, 1)
    \text{ is a monoid}
\end{align*}
\vspace*{-1cm}
\begin{align*}
    a\srmult (b\sradd c)
    &= (a\srmult b) \sradd (a\srmult c)
    & (b\sradd c) \srmult a
    &= (b \srmult a) \sradd (c\srmult a)\\
    a\srmult 0 &= 0\srmult a = 0  \ .
\end{align*}

A semiring comes with the so-called \bfemph{natural ordering} $\leq$ on its elements:  $a\leq a\sradd b$ for all $a, b\in \sr$. 
The semiring is called naturally ordered if $\leq$ is a partial order. 
In the following, we will determine suprema over sets of semiring elements.
These suprema will always be taken wrt. the natural ordering.

A naturally ordered semiring is \bfemph{$\omega$-continuous}
if it satisfies the following Properties~(1) to (4).
Property~(1) requires chains $(a_i)_{i\in \nat}$ to have a supremum $\sup\setcond{a_i}{i\in\nat}$ in $\sr$. 
Recall that a chain is a sequence with $a_i\leq a_{i+1}$ for all $i\in\nat$.
To state the Properties~(2) to~(4), given a sequence $(a_i)_{i\in \nat}$ in $\sr$ we define the infinite sum
\begin{align*}
    \bigsradd_{i\in \nat}a_i\ =\ \sup\setcond{a_0\sradd\ldots\sradd a_i}{i\in \nat}\ .
\end{align*}
Note that the sum exists by Property~(1). 
The Properties (2) to (4) now require
\begin{align*}
    c\srmult (\bigsradd_{i\in \nat}a_i) = \bigsradd_{i\in \nat} (c\srmult a_i)\qquad  (\bigsradd_{i\in \nat}a_i)\srmult c = \bigsradd_{i\in \nat} (a_i\srmult c)\qquad\bigsradd_{i\in \nat}a_i = \bigsradd_{i\in I}\bigsradd_{j\in J_i} a_j\ .
\end{align*}
The latter equality is supposed to hold for every partitioning 
of the natural numbers. 
The requirement for $\omega$-continuity allows us to define the \bfemph{Kleene-star} operator $^*:\sr\rightarrow \sr$ by $a^*=\bigsradd_{i\in \nat}a^i$ where we set $a^0=1$.

Throughout, we will work with idempotent and $\omega$-continuous semirings, \bfemph{io-semirings} for short.
A semiring is \bfemph{idempotent} if addition is idempotent, $a\sradd a=a$ for all $a\in \sr$. 
We will also consider the special case of \bfemph{commutative} io-semirings, where besides addition also  multiplication is commutative.  

Given a finite set of variables $\vars=\set{x_1,\ldots, x_k}$, a \bfemph{monomial} is an expression of the form
$m = a_1 \srmult x_{i_1} \srmult  \ldots \srmult a_l \srmult x_{i_l} \srmult a_{l+1}$,
where all $a_i \in \sr$ and $x_{i_j} \in \vars$.
A monomial without variables is a {\bf constant}. 
A \bfemph{polynomial} is a finite sum of monomials, $p = \bigsradd_{i = 1, \ldots, k} m_i$. 
We use $\polysr$ for the set of all polynomials.
A \bfemph{power series} is a countable sum of monomials.

The set of functions over $\sr$ with arguments $\vars$, denoted by $\fctring$, forms a semiring with element-wise addition and multiplication, i.e. for all $\vec a\in \sr^\vars$:
\begin{align*}
    (f \sradd g) (\vec a) = f(\vec a) \sradd g(\vec a)\qquad(f \srmult g) (\vec a) = f(\vec a) \srmult g(\vec a)\ .
\end{align*}
The \bfemph{semiring of functions} $\fctring$ is $\omega$-continuous (idempotent, commutative) if and only if $\sr$ is  $\omega$-continuous (idempotent, commutative).
We will also consider $\vars$-dimensional vectors of such functions, i.e. the set $(\fctring)^\vars$.
We can see such a vector as a single function $\sr^\vars \to \sr^\vars$.
Together with the component-wise operations, the set of functions $\sr^\vars \to \sr^\vars$ again forms a semiring, and the properties of $\sr$ carry over.
We are particularly interested in functions defined by polynomials.
Note that they are \bfemph{monotone}. 

Given a single function $f : \sr^\vars \to \sr $ and a vector of functions $\vec v : \sr^\vars \to \sr^\vars$, we write $\evalof{\vec v}{f}$ for the composition of $f$ and $\vec v$ defined by
\begin{align*}
    \evalof{\vec v}{f} : \sr^\vars \to \sr\qquad
    \evalof{\vec v}{f} (\vec a) = f( \vecat{v}{x_1} (\vec a), ..., \vecat{v}{x_k} (\vec a) ) \ .
\end{align*}
We also use evaluation on vectors of functions, 
$(\evalof{\vec v}{\vec f})_{x} = \evalof{\vec v}{\vec f_x}$. 
If $\vec a \in \sr^\vars$ is a vector of values, we moreover let $\evalof{\vec a}{f}$ denote the value of $f$ at $\vec a$.
Evaluation is monotone in both, the function and the argument, and it is associative.
Let $\vec a \leq \vec b$ be vectors of functions or values and let $f \leq g$ be functions. Let $\vec v$ be a vector of functions.
Then
\begin{align*}
    \evalof{\vec a}{f} \leq \evalof{\vec b}{f}
    \qquad
    \evalof{\vec a}{f} \leq \evalof{\vec a}{g}
    \qquad 
    \evalof{ \vec a} { \evalof{\vec v}{f}}
    = 
    \evalof{ \evalof{ \vec a} {\vec v} }{f}
    \ .
\end{align*}

Our contribution is a new method for solving systems of polynomial equations over $\sr$ in several unknowns $\vars$. 
A \bfemph{system of equations} is a vector of polynomials $\vec p$, which we denote by $\vec x = \vec p$. 
A \bfemph{solution} for it is a vector $\vec v \in \sr^\vars$ such that $\vec v = \eval{\vec v}(\vec p)$.
We always assume a vector of polynomials $\vec p$ to have components $\vecat{p}{x}$.


We will often relate languages $\lang\subseteq \Sigma^*$ over the alphabet $\Sigma=\sr\cup \vars$ of semiring elements and variables with functions.
If $\sr$ is infinite, we only need the finitely many elements that occur in the system of equations of interest.  
Given a word $w\in\Sigma^*$, we define $\srof{w}\in\polysr$ to be the monomial obtained by replacing concatenation with multiplication in the semiring.
Moreover, we define the function $\srof{\lang}=\bigsradd_{w\in\lang}\srof{w}$ summing up all monomials obtained from words in the language.
In turn, given a monomial $g\in\polysr$, we write $\wordof{g}\in \Sigma^*$ for the word obtained by understanding semiring multiplication as concatenation.



\section{Munchausen Iteration}
\label{Section:Acceleration}
We define the iteration scheme, relate it to derivation trees of context-free grammars, and with this relation derive our main theorem on convergence to the least fixed point, invoking deep results about Newton iteration from~\cite{EsparzaKieferLuttenberger2010}.
\subsection{Definition}
Consider a system of polynomial equations $\vec x = \vec p$.
Our method works over the semiring of functions $\sr^\vars \to \sr^\vars$. 
To highlight in $\vec p$ the functional aspect $\vec f$ and separate it from the constant part $\vec a$, we rewrite the system as $\vec x = \vec f + \vec a$. 
The first step is to compute a linear completion of the right-hand side polynomials $\vec f$.
The idea is borrowed from linear algebra, namely repeatedly substituting variables $y$ occurring in $\vecat{f}{x}$ by their defining functions $\vecat{f}{y}$.
To be able to represent the resulting function in a closed form, we focus on \bfemph{linear} substitutions where the next substitution is applied to a variable in $\vecat{f}{y}$. 

To render this formally, a \bfemph{substitution} is defined to be a pair consisting of a variable and a polynomial, denoted by $\set{x\mapsto g}$ from $\vars\times \polysr$. 
The \bfemph{application} of $\set{x\mapsto g}$ to a polynomial $f$ yields the set of polynomials $f\set{x\mapsto g}$ containing all variants of $f$ where one occurrence of $x$ has been replaced by $g$.
If $x$ does not occur in $f$, $f\set{x\mapsto g}$ is empty. 
We are interested in substitutions that are induced by the given system of equations and that are applied repeatedly, in the aforementioned linear fashion. 

\begin{definition}
    The set of \bfemph{linear polynomial substitutions} $\sigma_x$ for each variable $x \in \vars$ associated with $\vec f\in \polysr^{\vars}$ is defined
    mutually inductive
    by 
    \begin{align*}
        \sigma_x ::= \set{x \mapsto x} \bnf \set{x\mapsto g}\quad\text{with }g\in \vecat{f}{x}\sigma_y\ 
        .
    \end{align*}
    The set of \bfemph{linear monomial substitutions $\tau_x$} (for variable $x$) associated with $\vec f$ is defined similarly but works on monomials $m_{i_x}$ rather than the full $\vecat{f}{x}=\bigsradd_{i_x} m_{i_x}\ $.
\end{definition}
Linear polynomial substitutions yield the intuitive notion of completion we are aiming at.

\begin{definition}
    Consider $\vec f\in \polysr^{\vars}$. 
    Its \bfemph{linear completion} $\completionof{\vec f} \invfctring$ is defined component-wise by
    $\completionof{\vec f}_x = \completionof{\vecat{f}{x}} = \bigsradd_{\sigma_x} x\sigma_x$. 
    Here, $\sigma_x$ ranges over all linear polynomial substitutions $\sigma_x$ associated with $\vec f$.
\end{definition}
Linear monomial substitutions do not contain sums which makes them algorithmically easier to handle than the more general linear polynomial substitutions.
The following lemma shows that wrt. completion the two can be used interchangeably. 
The proof is by distributivity.

\begin{lemma}\label{Lemma:LinearMonomial}
    Consider $\vec f\in\polysr^{\vars}$. 
    Then 
    $\completionof{\vecat{f}{x}} = \bigsradd_{\tau_x} x \tau_x$.
\end{lemma}
Munchausen iteration starts with the linear completion.
The iteration step then inserts the current approximant into itself. 
Note that the method is applied only to the functional part $\vec f$ of the system of equations, and the result is a sequence of  functions $\munchit{n}$.
To obtain a value in $\sr$, we have to evaluate $\munchit{n}$ at some vector $\vec b \in \sr^\vars$.
An obvious choice for $\vec b$ is the given vector of constants $\vec a$. 
We will show that we can evaluate at any vector that lies between $\vec a$ and the least fixed point, and still converge to the least fixed point (Theorem \ref{Theorem:Convergence}).

\begin{definition}
Consider $\vec f\in \polysr^{\vars}$. The Munchausen iteration is
    \begin{align*}
        \munchit{0} \ =\ \completionof{\vec f}\qquad 
        \munchit{n+1}\ =\ \evalof{\munchit{n}}{\munchit{n}}\ .
    \end{align*}
The \bfemph{Munchausen sequence wrt.} $\vec b \in \sr^\vars$ is defined by $\munchsbit{n}
        \ = \
        \evalof{\vec b} {\munchit{n}}\ $.
\end{definition}
We aim to prove that the $\nth{n}$ element of the Munchausen sequence wrt. $\vec a$ is equal to the $\nth{2^n}$ Newton approximant.
We make the link to Newton iteration using derivation trees.


\subsection{Derivation Tree Analysis}
\label{subsec:CorrespondenceDerivationTrees}
Every system of polynomial equations relates naturally to a context-free grammar. 
We show that the $\nth{n}$ Munchausen approximant coincides with the yields of the derivation trees of dimension at most $2^n$ . 
For the development, consider the functional part ${\vec f\in\polysr^{\vars}}$ of the system of equations of interest. 
We associate with it the context-free grammar
\mbox{$\grammarof{\vec f}=(\vars, S, \bigcup_{x\in\vars}P_x)$.}
The variables form the non-terminals, the semiring elements give the terminal symbols. 
There is a set of production rules $P_x$ for every non-terminal $x$.
For the definition, assume $\vecat{f}{x}=\bigsradd_{i=1}^{k_x} m_i$. 
The productions are
\begin{align*}
P_x = \setcond{x\rightarrow \wordof{m_i}}{i=1,\ldots, k_x}\ .
\end{align*}
%
Since $\grammarof{\vec f}$ is a context-free grammar, we can make use of the concept of \bfemph{derivation trees} (that may have variables at the leaves). 
Let $\treesof{x}{}$ denote the set of all derivation trees that can be generated from the non-terminal $x$.
We write $\treesparof{x}{n}$ for the set of derivation trees in $\treesof{x}{}$ of dimension at most $n\in\nat$.
The \bfemph{dimension} $\dimof{\atree}$ of a tree $\atree$ is a well-known concept~\cite{Dimension2014} and defined inductively as follows. (i) If $\atree$ has no children, then $\dimof{\atree}=0$. \mbox{(ii) If} $\atree$ has precisely one child $\atree_1$, then $\dimof{\atree}=\dimof{\atree_1}$.
(iii) If $\atree$ has at least two children, consider the children $\atree_1$ and $\atree_2$ of highest dimension.
More formally, let $\dimof{\atree_1}\geq \dimof{\atree_2}$ and $\dimof{\atree_2}\geq \dimof{\atree'}$ for all children $\atree'\neq \atree_1$.
In this case, we set
\begin{align*}
    \dimof{\atree}\ =\
    \begin{cases}
        \dimof{\atree_1}+1  &\text{if }\dimof{\atree_1}=\dimof{\atree_2}\ ,\\
        \dimof{\atree_1}    &\text{if }\dimof{\atree_1}>\dimof{\atree_2}\ .\\
    \end{cases}
\end{align*}
The correspondence with derivation trees relies on the following lemma: 
A tree of dimension $2m$ can be decomposed into trees of dimension at most $m$. 
For the formal statement, if $t'$ is a tree with leaves $l_1, \ldots , l_n$ (from left to right) and $t_1, \ldots, t_n$ are trees, we denote the tree obtained by replacing each leaf $l_i$ with $t_i$ by $t' [t_1, \ldots, t_n]$ .
\begin{lemma}
\label{Lemma:Decomposition}
    Consider a tree $\atree$ with $\dimof{\atree}=2m$. 
    Then there are trees $\atree', \atree_1,\ldots, \atree_n$ with $\dimof{\atree'},\dimof{\atree_1},\ldots, \dimof{\atree_n}\leq m$ so that
    $\atree=\atree'[\atree_1,\ldots, \atree_n]$ .
\end{lemma}
\begin{proof}
    We identify the maximal subtrees $\atree_1$ to $\atree_n$ of $\atree$ that satisfy $\dimof{\atree_i}\leq m$. 
    Note that they are unique. 
    Removing them from $\atree$ leaves us with a subtree $\atree'$.
    Tree $\atree'$ has the same root as $\atree$.
    The leaves are labeled by $\rootof{\atree_1}$ to $\rootof{\atree_n}$. 
    If we replace each leaf $\rootof{\atree_i}$ by the tree $\atree_i$, we obtain a representation of $\atree$:
    \begin{align*}
    \atree\ =\ \atree'[\atree_1,\ldots, \atree_n]\ .
    \end{align*}
    To establish $\dimof{\atree'}\leq m$, assume towards a contradiction that $\dimof{\atree'}>m$.
    Consider the children $\atree_i$ that we removed from $\atree$ to obtain $\atree'$.
    By the maximality requirement for $\atree_i$, the parent node of $\rootof{\atree_i}$ in $\atree$ has outdegree $\geq 2$.
    
    Assume for every child $\atree_i$ one of the following holds: (i)
    $\dimof{\atree_i}=m$ or (ii) $t_i$ has a sibling of dimension $>m$ or (iii) $t_i$ has two siblings of dimension $=m$.
    In each of the cases, we can assume $t_i$ to contribute a dimension of $m$ when we determine $t'[t_1,\ldots, t_n]$. 
    As a result, we obtain
    \begin{align*}
    \dimof{\atree} = \dimof{\atree'}+m > m+m\ .
    \end{align*}
    This contradicts the assumption that the dimension of $\atree$ equals $2m$.
    
    As a consequence of this contradiction, there has to be a child $\atree_i$ that does not satisfy any of (i) to (iii) above.
    This means $\dimof{\atree_i}<m$, there is no sibling of dimension $>m$, and there is at most one sibling of dimension $m$.
    Let $x$ be the parent node of $\rootof{\atree_i}$. 
    Let $\atree_{x}$ be the subtree with $x$ as its root.
    The violation of (i) to (iii) allows us to conclude $\dimof{\atree_{x}}\leq m$.
    This in turn contradicts the maximality of $\atree_i$.
    
    Since we derived a contradiction in both cases (all children satisfy (i), (ii), or (iii) and there is a child that does not satisfy (i), (ii), and (iii)), we have to conclude that the assumption $\dimof{\atree'}>m$ has to be false.
\end{proof}
The correspondence is our first main result.
Recall that the yield of a given tree, $\yieldof{\atree}$, is the word formed by the leaves when the tree is traversed in left-first manner. 

\begin{theorem}
\label{Theorem:YieldDim}
    $\munchat{n}{x}\ =\ \bigsradd_{t \in \treesparof{x}{2^n}}\srof{\yieldof{t}}$ .
\end{theorem}
\begin{proof}
    We proceed by induction on $n$.\\[0.2cm]
    {\bf \sffamily Base case $n=0$}\quad  
    In the base case, we have 
    \begin{align*}
    \munchat{0}{x} =\completionof{\vecat{f}{x}}
    =\bigsradd_{\sigma_x} x \sigma_x
    =\bigsradd_{\tau_x} x \tau_x \ .
    \end{align*}
    The first two equalities are by definition, the last is Lemma~\ref{Lemma:LinearMonomial}. We have to show that 
    \begin{align*}
    \bigsradd_{\tau_x} x \tau_x
    = \bigsradd_{\atree \in \treesparof{x}{1}}\srof{\yieldof{\atree}}\ .
    \end{align*}
    To see that each $x \tau_x$ can be obtained as the yield of a derivation tree rooted in $x$, note that the substitution $\tau_x$ takes the form
    \begin{align*}
    \set{x\mapsto m_x\set{\ldots\set{z\mapsto m_z}}\ldots}\ .
    \end{align*}
    Here, $m_x$ to $m_z$ are monomials and so $x\rightarrow m_x$ up to $z \rightarrow m_z$ are production rules in the grammar $\grammarof{\vec f}$.
    Hence, the application of $\tau_x$ corresponds to a derivation sequence from~$x$.
    
    We show that the derivation tree has dimension at most one. 
    Since the substitution is linear, we can highlight in every rule the variable that will be replaced next.
    So in the initial step, we have $x\rightarrow \wordof{m_i}$ with $m_i = g_{1}\srmult y\srmult g_2$.
    In the derivation tree, the rule yields several subtrees for root $x$.
    Since $g_{1}$ and $g_2$ are monomials, their subtrees consist of single nodes labeled by an element from the semiring or a variable.
    The remaining subtree is for $y$ and the same reasoning applies.
    Since the subtrees labeled by a semiring element and the subtrees of a variable have dimension zero,
    the overall derivation tree has dimension at most one (zero, if no rule is applied).
    
    For the reverse direction, we have to show that for every subtree $\atree\in\treesof{x}$ of dimension at most one, we have a linear monomial substitution $\tau_x$ that we can apply to $x$ to obtain $\srof{\yieldof{\atree}}$. 
    The first observation is that in $\atree$, for every pair of siblings at least one has to have dimension zero. 
    Assume this was not the case and both siblings have dimension at least one. 
    In this case, their parent node has dimension at least two which contradicts the assumption on the dimension of $\atree$. 
    The only trees with dimension zero are linear paths. 
    Since there are no productions of the shape $x\rightarrow y$ in $\grammarof{\vec f}$
    (since we removed the constants from $\vec f$ and all other monomials are at least of the shape $a_1\srmult y\srmult a_2$) the trees of dimension zero have to be leaves.
    Combined with the fact that one sibling has to have dimension zero, we obtain that $\atree$ is a path with semiring elements and variables to the sides. Hence, it forms a linear monomial substitution. 
    Note that this covers the case where the path has dimension zero.
    Then the tree is $x$ itself, to which we apply $\set{x \mapsto x}$.\\[0.2cm]
    {\bf \sffamily Induction step}\quad
    Assume $\munchat{n}{x} =\bigsradd_{\atree \in \treesparof{x}{2^n}}\srof{\yieldof{\atree}}$ holds and consider $n+1$. \\[0.2cm]
    The following equations make use of the definition of $\munchit{n+1}$, the induction hypothesis, and the definition of evaluation:
    \begin{align*}
    \munchat{n+1}{x}
    &=\evalof{\munchit{n}}{\munchat{n}{x} }\\
    &=\evalof{\munchit{n}}{\bigsradd_{\atree\in\treesparof{x}{2^n}}  \srof{\yieldof{\atree}}\ }\\
    &= 
    \bigsradd_{\atree\in\treesparof{x}{2^n}}
    \evalof{\munchit{n}}{\ \srof{\yieldof{\atree}}\ }\ .
    \end{align*}
    The evaluation 
    $\evalof{\munchit{n}}{\ \srof{\yieldof{\atree}}\ }$
    replaces every variable $y$ in $\srof{\yieldof{\atree}}$ by $\munchat{n}{y}$.
    By the induction hypothesis, $\munchat{n}{y}=\bigsradd_{\atree\in\treesparof{y}{2^n}}\srof{\yieldof{\atree}}$.
    This means every variable $y$ is replaced by the sum of the yields of all derivation trees $\atree'$ with $\dimof{\atree'}\leq 2^n$. 
    By $\omega$-continuity, we can equivalently sum up all monomials that result from $\srof{\yieldof{\atree}}$ by replacing $y$ by the yield  of a single derivation tree $\atree'$.
    
    To establish the inequality $\munchat{n+1}{x} \leq \bigsradd_{\atree\in\treesparof{x}{2^{n+1}}} \srof{\yieldof{\atree}}$, note that 
    every monomial of $\evalof{\munchit{n}}{\ \srof{\yieldof{\atree}} \ }$ is obtained from a derivation tree $\atree''$ which equals $\atree$ but appends the trees $\atree'$ to the leaves.
    Since $\atree$ as well as the $\atree'$ have dimension at most $2^n$, the resulting tree $\atree''$ has dimension at most $2^n+2^n = 2^{n+1}$.
    
    For the reverse direction, we show $\srof{\yieldof{\atree}}\leq \munchat{n+1}{x}$.
    Consider a derivation tree $\atree\in{\treesof{x}}$ of dimension $2^{n+1}$.
    The same argumentation holds for trees of smaller dimension.
    By Lemma~\ref{Lemma:Decomposition}, the tree can be decomposed into $\atree'$ and $\atree_1,\ldots \atree_n$, all of dimension at most $2^n$:
    \begin{align*}
    \atree\ =\ \atree'[x_1\mapsto\atree_1, \ldots, x_n\mapsto\atree_n]\ .
    \end{align*}
    By definition of the yield, we get that $\yieldof{\atree}$ results from $\yieldof{\atree'}$ by replacing $x_1$ to $x_n$ with $\yieldof{\atree_1}$ to $\yieldof{\atree_n}$, respectively. 
    The above discussion concludes the case.
\end{proof}


\subsection{Results}
We prove that the Munchausen sequence converges to the least fixed point. 
Moreover, in the commutative case it is guaranteed to reach the least fixed point in a number of steps that is logarithmic in the number of variables.
Both results rely on a precise correspondence between Munchausen iteration and Newton iteration. 

To define the Newton iteration, we recall the concept of \bfemph{differentials}.  
The differential of a polynomial $p$ wrt. a variable $x \in \vars$ at point $\vec v$ is the polynomial defined inductively by
\begin{align*}
    \diffwrt{p}{\vec v}{x} =
    \begin{cases}
        \bigsradd_{i \in I} \diffwrt{m_i}{\vec v}{x}
        & \text{ if } p = \bigsradd_{i \in I} m_i\ , \\
        ( \diffwrt{g}{\vec v}{x} \srmult \evalof{\vec v}{h} )
        \sradd
        ( \evalof{\vec v}{g} \srmult \diffwrt{h}{\vec v}{x} )
        & \text{ if } p = g \srmult h\ , \\
        0 & \text{ if } p \in S \text{ or } p \in \vars \setminus \set{x}\ ,\\
        x & \text{ if } p = x\ .\\
    \end{cases}
\end{align*}
The differential of $p$ at point $\vec v$ is the sum   
$\diff{p}{\vec v} = \bigsradd_{x \in \vars} \diffwrt{p}{\vec v}{x}$ .
The differential of a vector of polynomials is defined component-wise, $(\diff{\vec p}{\vec v})_x = \diff{\vecat{p}{x}}{\vec v}$.
The function ${\diff{\vec f}{\vec v}}^*$ is defined by summing up all $i$-fold applications of the differential, i.e. 
${\diff{\vec f}{\vec v}}^*  = \bigsradd_{i \in \nat} {\diff{\vec f}{\vec v}}^i$ with
$
{\diff{\vec f}{\vec v}}^0 = \mathit{id}$ and 
${\diff{\vec f}{\vec v}}^{i+1} = \evalof{ {\diff{\vec f}{\vec v}}^{i}  }{ \diff{\vec{f}}{\vec v}}$\ .
 
With differentials at hand, the \bfemph{Newton iteration} is
\begin{align*}
    \newtonit{0}
    \ =\ \evalof{\vec 0}{\vec p}
    \qquad
    \newtonit{n+1}
    \ =\ \evalof
    {\newtonit{n}}
    { {\diff{\vec p}{\newtonit{n}}}^*}\ .
\end{align*}
Actually, this is not the most general definition of Newton iteration but coincides with it in the idempotent case that we consider.
An explanation of why the sequence mimics the classical method from numerics is beyond the scope of this paper. It can be found in~\cite{EsparzaKieferLuttenberger2010}.

The $\nth{n}$-Newton approximant is known to correspond to the derivation trees of dimension at most $n$.
To be precise, Esparza et al. consider \bfemph{complete} derivation trees where the yields do not contain variables.\footnote{To handle the non-idempotent case, the trees are also decorated. We elaborate on this in Section~\ref{Appendix:Decorated}.}
Let $G_{\vec f}(\vec a)$ be the grammar that adds to $G_{\vec f}$ the rules $x \to \vecat{a}{x}$ for each variable. 
Let $\ctreesparof{x}{n}$ denote the set of complete derivation trees of dimension at most $n$ from non-terminal $x$ in $G_{\vec f}(\vec a)$.

\begin{theorem}[Esparza et al. \cite{EsparzaKieferLuttenberger2010}]
\label{Theorem:Esparza}
    $
        {\newtonit{n}}_x\ =\ \bigsradd_{t \in \ctreesparof{x}{n} } \srof{\yieldof{t}}
        \ .
    $
\end{theorem}
We argue that the complete trees of dimension $n$ of $G_{\vec f}(\vec a)$ are precisely the (incomplete) trees of dimension $n$ of $G_{\vec f}$, extended by appending the constants.
Appending the constants means to every leaf labeled by $x$ we append a child node $\vecat{a}{x}$. 
To see the correspondence, note that removing or adding those appendices does not change the dimension. 
The semiring element corresponding to the yield of the extended tree is precisely the semiring element for the yield of the original tree evaluated at the vector $\vec a$.

\begin{lemma}
\label{Lemma:EvalVsComplete}
    $
        \evalof {\vec a} { \bigsradd_{t \in \treesparof{x}{2^n}}  \srof{\yieldof{t}} }
        \ =\ 
        \bigsradd_{t \in \ctreesparof{x}{2^n}}  \srof{ \yieldof{t}}\ .
    $
\end{lemma}
%
We can now show that the $\nth{n}$ element of the Munchausen sequence wrt. $\vec a$ equals the \mbox{$\nth{2^n}$ Newton} approximant. 
Since the Newton sequence converges to the least fixed point~$\mu\vec p$, so does the Munchausen sequence.
Evaluating at larger vectors $\vec b$ requires further arguments.

\begin{theorem}
\label{Theorem:Convergence}
    Let $\vec x = \vec p  = \vec f + \vec a$ be a system of polynomial equations.
    \begin{enumerate}[(1)]
        \item
            $\munchsit{n}= \newtonit{2^n}$ .
        \item
            Let $\vec a \leq \vec b \leq \mu \vec p$.
            Then
            $\sup_{n \in \nat} \munchsbit{n} = \mu \vec p$ .
    \end{enumerate}
\end{theorem}

\begin{proof}
    We show (1). 
    Using Theorem~\ref{Theorem:YieldDim}, Lemma~\ref{Lemma:EvalVsComplete}, and Theorem~\ref{Theorem:Esparza} yields
    \begin{align*}
        \munchsat{n}{x}
        = \evalof {\vec a} {\munchat{n}{x}}
        = \mathit{eval}_{\vec a} \big(\bigsradd_{t \in \treesparof{x}{2^n}}\srof{\yieldof{t}} \big )
        = \bigsradd_{t \in \ctreesparof{x}{2^n}}  \srof{ \yieldof{t}}
        = \newtonat{2^n}{x}
        \ .
    \end{align*}
\end{proof}
In the commutative case, we can apply another deep result from \cite{EsparzaKieferLuttenberger2010}: The number of iterations needed to reach the least fixed is at most the number of variables in $\vars$.

\begin{corollary}
    If $\sr$ is commutative, we have
    $\mu \vec p = \munchsit{ \lceil \log | \vars | \rceil } \ .$
\end{corollary}

\subsection{Related Methods}
We already elaborated on the relationship with Newton iteration and with Kleene iteration. 
An improvement of Newton iteration to a hierarchy (in terms of convergence speed) of iteration schemes appeared in~\cite{EKLSTACS2007}. 
The idea is to repeatedly apply the Newton operator to itself.
The main result shows that one application of the $n$-fold Newton operator and $n$ steps of Newton iteration coincide. 

The hierarchy of Newton iterations is substantially different from the Munchausen iteration we present here.
It relies on a linear derivation process that adds one dimension with each self application. 
In this (outer) derivation a result of dimension $n$ is inserted, leading to a result of dimension $n+1$.
Munchausen iteration inserts a derivation result of dimension $n$ into a derivation of dimension $n$, thus doubling the analysis information in every step.


\section{Algorithmic Considerations}
\label{Section:Algorithmics}
We study the operations of linear completion and evaluation as they are needed for the initial and for the iteration step of the Munchausen scheme.


\subsection{Linear Completion}
As indicated by the correspondence between the $\nth{0}$ Munchausen approximant and the $\nth{1}$ Newton approximant, the differential $D \vec f$ should be a possibility to represent the linear completion of $\vec f$.
To be precise, we need to sum up all $i$-fold applications of the differential to obtain the linear completion.
The proof shows that the $i$-fold application corresponds to all monomial substitutions of length $i+1$.
\begin{theorem}
    \label{Theorem:Differential}
    For every vector $\vec v \in \sr^\vars$, we have
    $\evalof{\vec v}{\completionof{\vec f}}  = \evalof{\vec v}{ {\diff{\vec f}{\vec v}}^*}$ .
\end{theorem}
We now show how to construct a linear context-free grammar that represents the linear completion. 
The benefit over Theorem~\ref{Theorem:Differential} is that we are not bound to using differentials but have available the spectrum of language-theoretic techniques --- even for regular languages (Section~\ref{Section:Tensor}). 
By Lemma~\ref{Lemma:LinearMonomial}, the linear completion is (for each variable) the sum 
\begin{align*}
\completionof{\vecat{f}{x}} = \bigsradd_{\tau_x } x \tau_x \qquad\text{where $\tau_x$ has the form}\qquad\tau_x = \set{x \mapsto m_{i_x}\set{y \mapsto{m_{i_y}\set{\ldots}}}}\ .
\end{align*}
By the definition of linear substitutions, after $x\mapsto m_{i_x}$ the next substitution $y\mapsto m_{i_y}$ will be applied to a single occurrence of $y$ in $m_{i_x}$. 
The idea of the grammar construction is to highlight in each monomial the variable that will be replaced next.
To be precise, we even fix the occurrence of the variable that will be rewritten.
Given a monomial $m$ and an occurrence $z$ of a variable in $m$, there are unique monomials $\lof{m}{z}$ and $\rof{m}{z}$ so that
\begin{align}
m=\lof{m}{z}\srmult z\srmult\rof{m}{z}\ .\label{Equation:DecompositionNonCommutative}
\end{align}
We define the grammar to be
${\lingrammarof{0}
    = ( 
        \setcond{\iter{y}{1}}{y \in \vars} ,
                \vars\cup\sr,
        \bigcup_{y \in \vars} P_y \cup P)
}$.
We create a non-terminal (with index) for each variable.   
The terminals are the variables and the semiring elements.
The reason $\lingrammarof{0}$ has non-terminals $\iter{y}{1}$ is that 
we will see an exponential growth in the number of non-terminals during evaluation when we make the grammars explicit (Section~\ref{Section:Evaluation}). 
Every monomial $m$ of $\vecat{f}{y}$ and every occurrence $z$ of a variable in $m$ will induce a rule
that mimics the decomposition in Equation~\eqref{Equation:DecompositionNonCommutative}. 
Note that all variables in $\lof{m}{z}$ and $\rof{m}{z}$ are terminals, which reflects the fact that they will not be replaced by further linear substitutions. 
Moreover, note that a variable may have several occurrences in $m$, in which case we obtain several rules:
\begin{align*}
   P_y =
   \setcond{
       \iter{y}{1} \to \wordof{\lof{m_{i_y}}{z}} \concat \iter{z}{1} \concat \wordof{\rof{m_{i_y}}{z}}
   }
   {
       \vecat{f}{y} = \bigsradd_{i_y} m_{i_y},
       z \text{ an occurrence in }m_{i_y}    
   }
   \ .
\end{align*}
The productions $P=\setcond{\iter{y}{1} \to y}{y \in \vars}$ 
mimic the identity substitution.
We obtain a one-to-one correspondence between the linear substitutions applied to $x$ and the sentential forms derivable from $\iter{x}{1}$, denoted by $\langof{\lingrammarof{0}_{x}}$. 

\begin{proposition}
\label{Proposition:GrammarCompletion}
    $\completionof{\vecat{f}{x}} = \srof{\langof{\lingrammarof{0}_{x}}}$
    .
\end{proposition}
Computing information from $\langof{\lingrammarof{0}_{x}}$ is still non-trivial since we do not have a closed expression for the language. 
There are two special cases when $\langof{\lingrammarof{0}_{x}}$ is easy to evaluate. 
If $\sr$ is finite, also the set of functions $\sr^\vars \to \sr$ is finite.
In this setting, a Kleene iteration applied to  $\lingrammarof{0}_{x}$ (more precisely, a system of linear equations obtained from the grammar) is sufficient to determine a closed-form description of the linear completion.

If $\sr$ is commutative, the grammar construction can be modified to ensure left-linearity.
Indeed, Equation~\eqref{Equation:DecompositionNonCommutative} simplifies to the following unique representation of a monomial $m$ wrt. a variable $z$ (we no longer have to work with variable occurrences): 
\begin{align}
    m = z\srmult m^z \ .\label{Equation:DecompositionCommutative}
\end{align}
This in turn simplifies the transitions to $\iter{y}{1} \to \iter{z}{1}\concat \wordof{m_{i_y}^{z}}$\ . 

The left-linear grammar yields a closed representation of the linear completion as a regular expression over $\vars \cup \sr$, on which further evaluation steps can be performed. 
Actually, we only need the Parikh image of the language~\cite{Parikh1966}, which is a semilinear set and potentially more compact. 

\subsection{Evaluation}\label{Section:Evaluation}
To capture $\munchit{n+1}$, we show how to reflect $\evalof{ \munchit{n} }{ \munchit{n}}$ on grammar level.
Assume we have a grammar $\lingrammarof{n}$ with language $\munchit{n}$. 
Our construction will maintain the invariant that $\lingrammarof{n}$ has non-terminals of the form $\iter{y}{m}$ with $1\leq m\leq 2^n$. 
The terminals will always be $\vars\cup\sr$.
The grammar for $\munchit{n+1}$ will behave like $\lingrammarof{n}$ but invoke itself when it reaches a terminal $y$.
To invoke $y$, we have to turn the variable into a non-terminal. 
We create two copies of $\lingrammarof{n}$ and modify the indices in one of the copies.
This index shift in particular turns a former terminal $y$ into $\iter{y}{2^n}$, which is a non-terminal in the other grammar:
\begin{align*}
\lingrammarof{n+1} = \lingrammarof{n}\cup (\lingrammarof{n}+2^n)\ . 
\end{align*}
Formally, the \bfemph{index shift by $k\in\nat$} turns $\lingrammarof{n}$ into the grammar $\lingrammarof{n}+k$, where consistently all non-terminal indices are increased by $k$ and all terminals $y$ are turned into non-terminals $\iter{y}{k}$. 
To give an example, the production $\iter{y}{i}\to a\cdot x\cdot \iter{z}{i}$ from $\lingrammarof{n}$ will be turned into $\iter{y}{i+k}\to a\cdot\iter{x}{k}\cdot\iter{z}{i+k}$ in $\lingrammarof{n}+k$.
The union of the grammars is taken componentwise.
Let the sentential forms derivable from $\iter{x}{2^n}$ be denoted by $\langof{\lingrammarof{n}_{x}}$.
\begin{proposition}
\label{Proposition:LinGrammar1}
    For each $n \in \nat$, we have $\munchat{n}{x} = \srof{\langof{\lingrammarof{n}_x}}$\ .
\end{proposition}
To get from the Munchausen iteration to the Munchausen sequence, we need to evaluate the function $\munchit{n}$ at a vector of constants $\vec b$. 
This operation can also be performed on the grammar.
We treat the occurrences of $y$ as non-terminals instead of terminals and add the rules $y \to \vecat{b}{y}$ for every  variable $y \in \vars$.
Let the resulting grammar be $\iter{LG(\vec b)}{n}$.
%
%
\begin{proposition}
\label{Proposition:LinGrammar2}
    For each $n \in \nat$, we have 
    $
        {\munchsbit{n}}_{x}\ =\ \srof{\langof{\iter{LG(\vec b)}{n}_x}}
        \ .
    $
\end{proposition}
The grammars $\lingrammarof{n}$ have productions of the same shape that only differ in the index~$n$.
We exploit this to give a more compact representation of the language by an \bfemph{indexed grammar}.
Indexed grammars annotate the non-terminals in the productions with a stack.

The indexed grammars $\indgram$ we define uses the same non-terminals and terminals as $\lingrammarof{0}$.
The stack $s \in 1^*0$ encodes the index in unary.
The set of production rules is $\bigcup_{y \in \vars } R_y \cup R$.
As in $P_y$, the productions in $R_y$ start in $\iter{y}{1}$ and single out one occurrence $\iter{z}{1}$ of a variable in a monomial of $\vecat{f}{y}$.
When using the rule, the stack $[ 1.s ]$ of $\iter{y}{1}$ is passed to $\iter{z}{1}$. 
Also the other variables are treated as non-terminals. 
For them, the stack height is decreased by one.
Formally, for each occurrence $z$ of a variable in a monomial $m_{i_y}$ of $\vecat{f}{y}$, the set $R_y$ has a rule
\begin{align*}
    \iter{y}{1} [ 1.s ] \to \wordof{\lof{m_{i_y}}{z}} [s]\ \concat\ \iter{z}{1} [1.s] \ \concat\ \wordof{\rof{m_{i_y}}{z}} [s]\ .
\end{align*}
%
The set $R$ contains a rule for each variable that replaces the non-terminal version by the terminal version if the stack is empty, $
    R = \setcond
    {
        \iter{y}{1}[ 0 ] \to y
    }
    {
        y \in \vars
    }
    \ .
$

We define $L(\iter{\indgram}{n}_x)$ to be the set of sentential forms derivable in $\indgram$ from $\iter{x}{1} [ 2^n]$, i.e. with the unary encoding of $2^n$ as initial stack content.
Obviously, \mbox{$L(\iter{\indgram}{n}_x) = L(\iter{LG}{n}_x)$}, and we can also perform the evaluation by adding rules as for $\iter{LG}{n}$.
This allows us to phrase the Propositions \ref{Proposition:LinGrammar1} and \ref{Proposition:LinGrammar2} in terms of the indexed grammar $\indgram$.

\subsection{Tensor Semirings}\label{Section:Tensor}
Left-linear grammars are preferrable over linear context-free ones for the better algorithmics they support (see below). 
We show that we can work with left-linear grammars also in the case of non-commutative io-semirings. 
To this end, we adapt the recent work~\cite{Reps2016}. 
Reps~et~al. have shown that --- 
provided the semiring of interest has an associated tensor-product semiring --- every system of linear equations over the semiring can be transformed to a left-linear system over the tensor-product semiring. 
One important example where a tensor-product semiring exists is predicate abstraction~\cite{Reps2016}. 

\begin{definition}
    We call an io-semiring $\sr$ \bfemph{admissible}, if there is a transpose operation, an associated tensor-product semiring, and a readout operation.

    The \bfemph{transpose} $\cdot^t : \sr \to \sr$ should satisfy
    \begin{align*}
        (a \sradd b)^t &= a^t \sradd b^t
        & (a \srmult b)^t &= b^t \srmult a^t
        & (a^t)^t &= a
        \ .
    \end{align*}
    A \bfemph{tensor-product semiring} $\tensor{\sr}$ is an 
    io-semiring $(\tensor{\sr}, \tadd, \tmult, \tensor{0}, \tensor{1})$ together with a map $\tp : \sr \times \sr \to \tensor{\sr}$
    such that
    \begin{align*}
        0 \tp a = a \tp 0 &= \tensor{0}
        & (a \tp b) \tmult (c \tp d) &= (a \srmult c) \tp (b \srmult d)\\
        a \tp (b \sradd c) &= (a \tp b) \tadd (a \tp c)
        & (b \sradd c) \tp a &= (b \tp a) \tadd (c \tp a)
        \ .
    \end{align*}
    The \bfemph{readout} operation $\readout : \tensor{\sr} \to \sr$ should satisfy (with $I$ finite or countable)
    \begin{align*}
        \roof { a \tp b } &= a^t \srmult b
        & \roof { \bigsradd_{i \in I} p_i } &= \bigsradd_{i \in I} \roof{p_i}\ .
    \end{align*}
    %
\end{definition}
The crucial requirement is the existence of a readout operation that distributes over sums without producing cross terms.
It is, for example, not met by the language semiring.

Consider a system of linear equations over an admissible semiring $\sr$ of the form
\begin{align}
\label{EquationSysLinear}
    x_i &= c_i \sradd
    \underset{j = 1, \ldots, k}{\bigsradd} a_{i,j} \srmult x_j \srmult b_{i,j}
    \quad \text{ for } i = 1, \ldots, k\ .
\end{align}
Reps et al. define its \bfemph{regularization} to be the left-linear system over the associated tensor-product semiring $\tensor{\sr}$:
\begin{align}
\label{EquationSysLeftLinear}
    y_i &= ( 1^t \tp c_i)  \tadd 
    \underset{j = 1, \ldots, k}
    {
        \mathlarger
        \textstyle{\bigsradd}_{\mathcal{T}}
    } 
    y_j \tmult (a_{i,j}^t \tp b_{i,j})
    \quad \text{ for } i = 1, \ldots, k\ .
\end{align}
Their main result shows that the least solution to (\ref{EquationSysLinear}) can be obtained from the least solution to (\ref{EquationSysLeftLinear}) by applying the readout operation.
\begin{theorem}[Reps et al. \cite{Reps2016}]
\label{Theorem:Regularization}
    Let $\vec v$ be the least solution to (\ref{EquationSysLeftLinear}).
    Then $\roof{\vec v}$ is the least solution to (\ref{EquationSysLinear}).
\end{theorem}
The importance of the result stems from the fact that systems of  left-linear equations (\ref{EquationSysLeftLinear}) enjoy efficient algorithmics.
For example, Tarjan's path-expression algorithm~\cite{Tarjan1981} can be applied to (\ref{EquationSysLeftLinear}) to obtain for every $y_i$ a regular expression (over the tensor-product semiring) capturing the least solution.
We discuss how to use this in our setting.

Consider the system of linear equations for the linear completion that is obtained from $\lingrammarof{0}$. 
Let $\tunchit{0}$ denote its regularization.
With Tarjan's algorithm, we obtain for $\tunchit{0}$ a regular expressions over the tensor semiring and $\vars$. 
As a consequence of Theorem~\ref{Theorem:Regularization} and Proposition~\ref{Proposition:GrammarCompletion}, we have $\completionof{\vec f} =  \roof{\tunchit{0}}$.

One would also like to carry out the evaluation process over the tensor-product semiring.
Unfortunately, $\tunchit{0}$ is a regular expression with variables denoting elements from $\sr$, namely those occurrences that were treated as terminals by the grammar. 
Therefore, we cannot evaluate $\tunchit{0}$ at $\tunchit{0}$, but only at $\roof{ \tunchit{0}}$.
We define $\tunchit{n+1} = \evalof{ \roof{ \tunchit{n} }} { \tunchit{n}}$. 
Using Theorem \ref{Theorem:Regularization} and induction, we get 
$\munchit{n} = \roof{ \tunchit{n}   }$
for all $n \in \nat$.
For an implementation, the idea would be to nevertheless insert the tensor element $\tunchit{0}$ and define a recursive readout.



\section{Discussion}\label{Section:Discussion}
We gave a new iteration scheme for solving polynomial equations over $\omega$-continuous and idempotent semirings.
The key idea is to solve the equations over the semiring of functions rather than the semiring of interest and only evaluate the resulting function when needed.
We showed that the method is exponentially faster than the well-known Newton sequence~\cite{EsparzaKieferLuttenberger2010}, and that we can obtain symbolic descriptions for the solutions.
The descriptions can be understood as identifying maximal sharing in the derivation trees of context-free grammars. 

Unfortunately, we do not yet know how to handle these descriptions.
If we give them explicitly as linear context-free grammars, semilinear sets, or regular expressions over the tensor semiring, the size of the description doubles in every step. 
Hence, we buy an exponential improvement in time at the cost of an exponential blow up in space.
This still means we compute a description of size $n$ in $\mathit{log}\ n$ steps. 
Experiments will have to tell how this compares to Newton iteration that, for the same result, needs $n$ steps 
but where the objects are semiring values rather than grammars.

The descriptions we obtain are structured to an extent that allows us to represent them symbolically, by a restricted class of indexed grammars (over linear context-free grammars, semilinear sets, or regular expressions).
With restricted indexed grammars, the iteration steps of Munchausen are easy to compute.
The drawback is that, so far, we do not know how to extract information from the restricted indexed grammars.
As future work, we plan to understand how to compute in such highly symbolic structures.

\section*{Acknowledgments}
We thank Stefan Kiefer for helpful discussions.

\bibliography{cited}

\newpage

\appendix


\section{Proofs of Section \ref{Section:Acceleration}}

\begin{proof}[Proof of Lemma \ref{Lemma:LinearMonomial}]
    Inequality $\geq$ is immediate by the fact that for every linear monomial substitution, there is a linear polynomial substitution that always inserts the full polynomial rather than one of its monomials.
    To show $\leq$, assume that every substitution in $\sigma_x$ applies to precisely one position so that $f \sigma_x$ is a single function and not a set.
    The general case follows since we do not make an assumption of where the application occurs.
    Substitution $\sigma_x$ has the shape
    \begin{align*}
    \set{
        x\mapsto \vecat{f}{x}
        \set{
            y \mapsto \vecat{f}{y}
            \set{
                \ldots
                \set{
                    z\mapsto \vecat{f}{z}
                }
                \ldots
            }
        }
    }\ .
    \end{align*}
    If we assume $\vecat{f}{x} =\bigsradd_{i_x}m_{i_x}$ and similarly for $\vecat{f}{y}$ and $\vecat{f}{z}$, then $x \sigma_x$ is of the form
    \begin{align*}
    \bigsradd_{i_x \neq i} m_{i_x}
    \sradd 
    \Big(
    m_{i, 1} \srmult
    \Big(
    \bigsradd_{i_y \neq j} m_{i_y}\sradd
    \big(
    m_{j, 1}\srmult
    (
    \ldots \srmult (\bigsradd_{i_z} m_{i_z}) \srmult \ldots 
    )
    \srmult m_{j, 2}
    \big)
    \Big)
    \srmult m_{i, 2}
    \Big)\ .
    \end{align*}
    Distributivity yields
    \begin{align*}
    \bigsradd_{i_x \neq i} m_{i_x} \sradd 
    \bigsradd_{i_y\neq j} (m_{i, 1}\srmult m_{i_y}\srmult m_{i, 2})\sradd\ldots \sradd
    \bigsradd_{i_z} (m_{i, 1}\srmult m_{j, 1}\srmult \ldots \srmult m_{i_z}\srmult \ldots \srmult m_{j, 2}\srmult m_{i, 2})\ .
    \end{align*}
    Each of these monomials is obtained by applying a linear monomial substitution to $x$, ranging from the identity substitution $\set{x \mapsto x}$ (not shown) over inserting some $m_{i_x}$ from $\vecat{f}{x}$ to a long substitution that ends with a monomial $m_{i_z}$. 
    Since all these substitutions are covered by the sum over $\tau_x$, we obtain the desired inequality.
\end{proof}

\begin{proof}[Proof of Theorem~\ref{Theorem:Convergence}(2)]
    By monotonicity of $\mathit{eval}$ in the evaluation point, we obtain the inequality
    \mbox{$\evalof{\vec b}{\munchit{n}} \geq \evalof{\vec a}{\munchit{n}}$}
    for all $n \in \nat$.
    Together with Theorem~\ref{Theorem:Convergence}(1), this yields
    \begin{align*}
    \sup_{n \in \nat} \evalof{\vec b}{\munchit{n}}
    \geq
    \sup_{n \in \nat} \evalof{\vec a}{\munchit{n}}
    = \mu \vec p\ . 
    \end{align*}
    
It remains to establish $\sup_{n \in \nat} \evalof{\vec b}{\munchit{n}} \leq \mu \vec p$.
We show that this inequality holds for $\vec b = \mu \vec p$. 
The desired statement then follows by monotonicity. 
In fact, it is enough to show that \mbox{$\evalof{\mu \vec p}{\munchit{n}} \leq \mu \vec p$} for all $n\in\nat$.
    We proceed by induction on $n$.
    
    In the base case, we establish 
    \mbox{$
        \evalof{ \mu \vec p}{ \completionof{\vecat{f}{x}}} \leq (\mu \vec p)_x
    $}
    simultaneously for all variables.
    Using Lemma~\ref{Lemma:LinearMonomial}, it sufficient to prove
    \mbox{$\evalof{ \mu \vec p}{  \bigsradd_{\tau_x} x\tau_x  } \leq (\mu \vec p)_x$}, where $\tau_x$ ranges over all linear monomial substitutions for $x$.
    
    We first show that for all $x$, for all monomials $m_i$ of
    $\vecat{f}{x} = \bigsradd_{i_x} m_{i_x}$,
    and for all $\tau_y$ where $y$ occurs in $m_i$, 
    we have that 
    $\evalof{ \mu \vec p}{g} \leq (\mu \vec p)_x$ for every
    $g \in m_i \tau_y$. 
    We proceed by an induction on the structure of $\tau_y$.
    In the base case, $\tau_y = \set{ y \mapsto y}$, $g = m_i$, and thus
    \begin{align*}
        \evalof{ \mu \vec p}{g}
        &= \evalof{ \mu \vec p}{m_i} \\
        &\leq \evalof{ \mu \vec p}{m_i}
        \sradd \evalof{ \mu \vec p}{ \bigsradd_{i_x \neq i} m_{i_x}}
        \sradd \vecat{a}{x}
        = \evalof{ \mu \vec p}{\vecat{p}{x}}
        = (\mu \vec p)_x
        \ .
    \end{align*}
    The last equality holds as $\mu \vec p$ is a fixed point of $\vec p$.
    Assume $\tau_y = \set{ y \mapsto h}$ with $h \in m_{i_y} \tau_z$.
    Let furthermore $m_i = m_1 \srmult y \srmult m_2$ so that $y$ is replaced to obtain $g$. 
    We have
    \begin{align*}
        \evalof{ \mu \vec p}{g}
        &= \evalof{ \mu \vec p}{m_1 \srmult h \srmult m_2 } \\
        &= \evalof{ \mu \vec p}{m_1}
        \srmult
        \evalof{ \mu \vec p}{h}
        \srmult
        \evalof{ \mu \vec p}{m_2}\\
        &\leq \evalof{ \mu \vec p}{m_1}
        \srmult
        (\mu \vec p)_y
        \srmult
        \evalof{ \mu \vec p}{m_2}\\
        &=  \evalof{ \mu \vec p}{m_1}
        \srmult
        \evalof{ \mu \vec p}{y}
        \srmult
        \evalof{ \mu \vec p}{m_2}\\
        &= \evalof{ \mu \vec p} {m_1 \srmult y \srmult m_2}\\
        &= \evalof{ \mu \vec p} {m_i}
        \leq (\mu \vec p)_x 
        \ .
    \end{align*}
    The first inequality is by the induction hypothesis combined with monotonicity, the second inequality is proven in the base case.
    
We can now derive
    $\evalof{ \mu \vec p}{  \bigsradd_{\tau_x} x\tau_x  } \leq (\mu \vec p)_x$
    by showing
    $\evalof{\mu \vec p}{x\tau_x} \leq (\mu \vec p)_x$
    for all $\tau_x$ and by using idempotence.
    If $\tau_x = \set{ x \mapsto x}$, we have
    $\evalof{\mu \vec p}{x\tau_x}  = \evalof{\mu \vec p}{x} = (\mu \vec p)_x$.
    If the substitution is $\tau_x = \set{ x \mapsto g}$ with $g \in m_{i_x} \tau_y$, we use the statement proven above to conclude
    $\evalof{\mu \vec p}{x\tau_x} = \evalof{ \mu \vec p}{g} \leq (\mu \vec p)_x$.
    
    Let us now assume that the statement holds for $n$.
    By definition and associativity, we get
    $
    \evalof{\mu \vec p}{ \munchit{n+1} }
    =
    \evalof{\mu \vec p}{ \evalof{\munchit{n}}{\munchit{n} }}
    =
    \evalof{ \evalof{\mu \vec p}{\munchit{n}} }{\munchit{n}}
    $.
    Using the induction hypothesis together with monotonicity, this is at most $\evalof{\mu \vec p}{\munchit{n}}$.
    Applying the induction hypothesis again yields the desired inequality.
\end{proof}


\section{Proofs of Section \ref{Section:Algorithmics}}

\begin{proof}[Proof of Theorem \ref{Theorem:Differential}]
    We establish 
    $
    \evalof{\vec v}{\completionof{\vecat{f}{x}}} 
    =
    \evalof{\vec v}{
        ({\diff{\vec f}{\vec v}}^*)_x
    }
    $
    simultaneously for all components. 
    Using Lemma~\ref{Lemma:LinearMonomial}, we have
    $\completionof{\vecat{f}{x}} = \bigsradd_{\tau_x} x \tau_x$,
    where $\tau_x$ ranges over all linear monomial substitutions for $x$.  
    Recall the definitions
    \begin{numcases}{ \diffwrt{p}{\vec v}{x} = }
    \bigsradd_{i \in I} \diffwrt{m_i}{\vec v}{x} 
    & $\text{ if } p = \bigsradd_{i \in I} m_i$
    \ , \label{differential1} \\
    ( \diffwrt{g}{\vec v}{x} \srmult \evalof{\vec v}{h} )
    \sradd
    ( \evalof{\vec v}{g} \srmult \diffwrt{h}{\vec v}{x} )
    & $\text{ if } p = g \srmult h$
    \ , \label{differential2} \\
    0 & $\text{ if } p \in S \text{ or } p \in \vars \setminus \set{x}$
    \ , \label{differential3} \\
    x & $\text{ if } p = x$
    \ \label{differential4} .
    \end{numcases}
    and
    \begin{align*}
    {\diff{\vec f}{\vec v}}^* &= \bigsradd_{i \in \nat} {\diff{\vec f}{\vec v}}^i,
    \quad \text{ where}
    &{\diff{\vec f}{\vec v}}^0 &= \mathit{id},
    &{\diff{\vec f}{\vec v}}^{i+1} &= \evalof{ {\diff{\vec f}{\vec v}}^{i}  }{ \diff{\vec{f}}{\vec v}}
    \ .
    \end{align*}
    We start by proving  $\leq$.
    First note that for $\tau_x = \set{x \mapsto x}$ with $x \tau_x = x$ we also have the summand
    $({\diff{\vec f}{\vec v}}^0)_x = \mathit{id}_x = x$ in ${\diff{\vec{f}}{\vec v}}^*$. 
    In general, by summand we mean a part of the sum that forms the differential. 
    To complete this part of the proof, we show by induction that for every $m_x \tau_y$, there is an $i$ and a summand $s$ of
    ${\diff{\vecat{f}{x}}{\vec v}}^i$ such that they evaluate to the same result under $\vec v$.
    Let us write $m_x = m_1 \srmult y \srmult m_2$, where $y$ is the occurrence that will be replaced by $\tau_y$.
    
    In the base case, let $\tau_y = \set{ y \mapsto y}$ and thus $m_x \tau_y = m_x$.
    Recall that
    \mbox{$({\diff{\vec f }{\vec v}}^1)_x = \diff{\vecat{f}{x}}{\vec v}$}
    is defined by summing up the differentials with respect to the single variables.
    We consider the differential with respect to variable $y$ and the summand that we get by selecting \mbox{monomial $m_x$} (Part (\ref{differential1}) of the Definition).
    This summand itself is a sum obtained by the application of the product rule (Part (\ref{differential2})) to $m_x$.
    Note that fully unfolding the product rule means that the base case (Parts (\ref{differential3}) and \ref{differential4})) is applied to one single symbol in $m_x$, and all other symbols are evaluated at $\vec v$.
    We consider the summand
    \mbox{$s = \evalof{\vec v}{m_1} \srmult \diffwrt{y}{\vec v}{y} \srmult \evalof{\vec v}{m_2}$}
    that is obtained by evaluating all symbols but $y$.
    The differential of $y$ with respect to $y$ is again $y$, so we get 
    \mbox{$s = \evalof{\vec v}{m_1} \srmult y \srmult \evalof{\vec v}{m_2}$.}
    This shows that the summand is evaluated to
    \begin{align*}
    \evalof{ \vec v} {\evalof{\vec v}{m_1} \srmult y \srmult \evalof{\vec v}{m_2}}
    &=
    \evalof{\vec v}{m_1} \srmult \evalof{ \vec v}{y} \srmult \evalof{\vec v}{m_2}\\
    &=
    \evalof{\vec v}{m_1 \srmult y \srmult m_2 }\\
    &=
    \evalof{\vec v}{m_x}
    \ .
    \end{align*}
    Let us now consider $\tau_y = \set{ y \mapsto g}$, with $g \in m_y \tau_z$ (where $m_y$ is a monomial of $\vecat{f}{y}$).
    By induction, there is an $i$ and a summand $s'$ of $({\diff{\vec f}{\vec v}}^i)_y$ such that evaluating $s'$ and $g$ leads to the same result.
    We look at $({\diff{\vec f}{\vec v}}^{i+1})_x = \evalof{ {\diff{\vec f}{\vec v}}^{i}  }{ \diff{\vecat{f}{x}}{\vec v}}$.
    We consider each summand
    \mbox{$s = \evalof{\vec v}{m_1} \srmult y \srmult \evalof{\vec v}{m_2}$}
    of $\diff{\vecat{f}{x}}{\vec v}$ as in the base case.
    Evaluating this summand at ${\diff{\vec f}{\vec v}}^{i}$ will evaluate $y$ to the sum
    $({\diff{\vec f}{\vec v}}^{i})_y$ containing $s'$.
    Using distributivity yields a new sum containing the summand given by evaluating $s$ at the summand $s'$ of ${\diff{\vecat{f}{x}}{\vec v}}^{i}$.
    Using the assumption that $s$ evaluates to $g$, this evaluates to
    \begin{align*}
    \evalof{ \vec v} {\evalof{\vec v}{m_1} \srmult s' \srmult \evalof{\vec v}{m_2}}
    &=
    \evalof{\vec v}{m_1} \srmult \evalof{ \vec v}{s'} \srmult \evalof{\vec v}{m_2}\\
    &=
    \evalof{\vec v}{m_1} \srmult \evalof{ \vec v}{g} \srmult \evalof{\vec v}{m_2}\\
    &=
    \evalof{\vec v}{m_1 \srmult g \srmult m_2}\\
    &=
    \evalof{\vec v}{m \tau_y}
    \ .
    \end{align*}
    To show $\geq$, we argue that summands of the polynomial defining $({\diff{\vec{f}}{\vec v}}^{i})_x$ correspond to substitutions applied to $x$.
    For $({\diff{\vec{f}}{\vec v}}^{0})_x = id_x$, we can select the substitution $\set{ x \to x}$.
    
    We will show that for any  $i > 0$ and any summand $s$ of $({\diff{\vecat{f}{x}}{\vec v}}^{i})_x$, there is a monomial $m_x$ of $\vecat{f}{x}$, a substitution $\tau_y$ and $g \in m_x \tau_x$ such that $s$ and $g$ evaluate to the same result.
    
    In the base case, note that $({\diff{\vecat{f}{x}}{\vec v}}^1)_x = {\diff{\vecat{f}{x}}{\vec v}}$ is a sum of the $\diffwrt{\vecat{f}{x}}{\vec v}{y}$ for all $y \in \vars$.
    Let us fix some $y$, then $\diffwrt{\vecat{f}{x}}{\vec v}{y}$ is a sum with the summands corresponding to the monomials of $\vecat{f}{x}$ (Part (\ref{differential1})).
    If some monomial $m_x$ does not contain $y$, all unfoldings of the product rule will have $0$ as a factor.
    Analogously, unfolding the product rule such that the differential is applied to a symbol other than $y$ will result in $0$.
    Let us fix an unfolding of the product rule not resulting in $0$, and let $m_x = m_1 \srmult y \srmult m_2$ be the corresponding decomposition of $m_x$.
    The corresponding summand of $\diffwrt{\vecat{f}{x}}{\vec v}{y}$ is $\evalof{\vec v}{m_1} \srmult y \srmult \evalof{\vec v}{m_2}$.
    With an argumentation analogous to the one used in the first part of the proof, this evaluates just as the element of $m_x \set{y \to y}$ does, where the substitution is applied to the occurrence of $y$ as in the decomposition.
    
    Now let us consider  $({\diff{\vec f}{\vec v}}^{i+1})_x = \evalof{ {\diff{\vec f}{\vec v}}^{i}  }{ \diff{\vecat{f}{x}}{\vec v}}$.
    Every summand of $({\diff{\vec f}{\vec v}}^{i+1})_x$ corresponds to evaluating a summand $s$ of ${\diff{\vecat{f}{x}}{\vec v}}$ at
    ${\diff{\vec f}{\vec v}}^{i}$.
    As in the base case, we can assume that $s$ has shape $\evalof{\vec v}{m_1} \srmult y \srmult \evalof{\vec v}{m_2}$.
    Evaluating $s$ will result in
    \mbox{$\evalof{\vec v}{m_1} \srmult ({\diff{\vec f}{\vec v}}^{i})_y \srmult \evalof{\vec v}{m_2}$.}
    Using distributivity, we get a large sum in which one single summand corresponds to evaluating $s$ at a summand $s'$ of $({\diff{\vec f}{\vec v}}^{i})_y$.
    By induction, there is a monomial $m_y$ of $\vecat{f}{y}$, a substitution $\tau_z$ and $h \in m_y \tau_z$ such that $s'$ and $h$ evaluate to the same result.
    We consider the substitution $\tau_y = \set{ y \mapsto h}$, and $g \in m_x \tau_y$, where the substitution applies to the same occurrence of $y$ as in $s$.
    Using that $s'$ and $h$ evaluate to the same result, we get
    \begin{align*}
    \evalof{\vec v} { \evalof{\vec v}{m_1} \srmult s' \srmult \evalof{\vec v}{m_2} }
    &=
    \evalof{\vec v}{m_1} \srmult \evalof{\vec v}{s'} \srmult \evalof{\vec v}{m_2}\\
    &=
    \evalof{\vec v}{m_1} \srmult \evalof{\vec v}{h} \srmult \evalof{\vec v}{m_2}\\
    &=
    \evalof{\vec v}{m_1 \srmult h \srmult m_2}\\
    &=
    \evalof{\vec v}{g}
    \ .
    \end{align*}  
\end{proof}


\section{Decorated Derivation Trees}
\label{Appendix:Decorated}

The nodes in the derivation trees from \cite{EsparzaKieferLuttenberger2010} are \bfemph{decorated}: They are not only labeled by a symbol, but also by the rule that was used to derive the symbol.
Let $\dtreesparof{x}{n}$ denote the set of all decorated complete derivation trees of dimension at most $n$.
One derivation tree as defined in our setting might correspond to several derivation trees with different additional labels.
We obtain each tree in $\ctreesparof{x}{n}$ by projecting all labels of a tree in $\dtreesparof{x}{n}$ to the first component,
and every tree in $\dtreesparof{x}{n}$ can be projected to a tree in $\ctreesparof{x}{n}$. 
Since the $\mathit{yield}$ function ignores the additional labels and since we assume idempotence, we end up with the same result if we sum up the yields of all undecorated trees:
\begin{align*}
    \bigsradd_{t \in \dtreesparof{x}{n} } \yieldof{t}
    &= \bigsradd_{t \in \ctreesparof{x}{n} }
    \bigsradd_{\substack{t' \in \dtreesparof{x}{n}, \\ \mathit{project}(t') = t}}
    \yieldof{t'}\\
    &=
    \bigsradd_{t \in \ctreesparof{x}{n} }
    \bigsradd_{\substack{t' \in \dtreesparof{x}{n}, \\ \mathit{project}(t') = t}}
    \yieldof{t}\\
    &= \bigsradd_{t \in \ctreesparof{x}{n} } \yieldof{t}
    \ .
\end{align*}


\section{The Non-Idempotent Case}
One may ask whether the Munchausen sequence also converges to the least fixed point in the case when the underlying semiring is not idempotent. 
The proof of Theorem~\ref{Theorem:YieldDim} does not hold in the non-idempotent case: The trees of dimension lower than $2^{n+1}$ are summed up several times. (For example, a tree of dimension $2^{n}+1$ may occur in the sum as a list of trees of dimension $1$ plugged into a tree of dimension $2^n$ and as a list of trees of dimension $2^n$ plugged into a tree of dimension $1$).
This problem cannot be solved by considering decorated derivation trees as in~\cite{EsparzaKieferLuttenberger2010}.
Even if we distinguish derivation trees that have the same shape but were created using different rules, the sum might contain multiple occurrences of one decorated derivation tree.
Therefore, in the non-idempotent case, the convergence results for Newton iteration do not carry over to Munchausen iteration.

We demonstrate that indeed Munchausen iteration may compute values strictly larger than the least fixed point in the following example.

\begin{example}
    Consider the following system of equations over the commutative but not idempotent $\omega$-continuous semiring of natural numbers with infinity $(\nat \cup \set{\infty}, +, \cdot, 0,1 )$:
    \begin{align*}
        x &= y \cdot y\
        &y &= z
        &z &= 2
        \ .
    \end{align*}
    Applying the decomposition into the non-constant and the constant part, we may write it as $\vec x = \vec f + \vec a = (y \cdot y, z, 0) + (0, 0, 2)$.
    Its linear completion is
    \begin{align*}
        \completionof{\vecat{f}{x}} &= x + y \cdot y + z \cdot y + y \cdot z
        &\completionof{\vecat{f}{y}} &= y + z
        &\completionof{\vecat{f}{z}} &= z
        \ .
    \end{align*}
    Evaluating it at the vector of constants $(0, 0, 2)$ yields $\evalof{\vec a}{ \completionof{\vec f} } =  (0 , 2, 2)$.
    Plugging in the linear completion into itself to obtain the $\nth{1}$ Munchausen approximant yields
    \begin{align*}
        \munchat{1}{x} &= (x + y \cdot y + z \cdot y + y \cdot z) + (y + z) \cdot (y + z) + z \cdot (y + z) + (y + z) \cdot z\\
        \munchat{1}{y} &= (y + z) + z\\
        \munchat{1}{z} &= z \ .
    \end{align*}
    We get $\munchsit{1} = (12,4,2)$, which is already strictly larger than the least fixed point $(4,2,2)$.    
\end{example}

\end{document}